\documentclass[12pt,a4paper,reqno]{amsart}
\usepackage{amsmath,amssymb,graphics,epsfig,color,enumerate,mathrsfs}
\usepackage{dsfont}  \usepackage{verbatim}
\definecolor{refkey}{gray}{.75}
\definecolor{labelkey}{gray}{.5}

\newtheorem{Theorem}{Theorem}[section]

\newtheorem{Lemma}[Theorem]{Lemma}
\newtheorem{Proposition}[Theorem]{Proposition}
\newtheorem{Corollary}[Theorem]{Corollary}
\newtheorem{Remark}[Theorem]{Remark}

\newtheorem{Definition}[Theorem]{Definition}

 \definecolor{darkgreen}{rgb}{0,0.4,0}

\definecolor{light}{gray}{0.9}

\newcommand{\cC}{\ensuremath{\mathcal C}}

\newcommand{\cE}{\ensuremath{\mathcal E}}

\newcommand{\cG}{\ensuremath{\mathcal G}}
\newcommand{\cH}{\ensuremath{\mathcal H}}

\newcommand{\cN}{\ensuremath{\mathcal N}}

\newcommand{\cV}{\ensuremath{\mathcal V}}

\newcommand{\bbE}{{\ensuremath{\mathbb E}} }

\newcommand{\bbI}{{\ensuremath{\mathbb I}} }

\newcommand{\bbN}{{\ensuremath{\mathbb N}} }

\newcommand{\bbP}{{\ensuremath{\mathbb P}} }

\newcommand{\bbR}{{\ensuremath{\mathbb R}} }

\newcommand{\bbZ}{{\ensuremath{\mathbb Z}} }

\newcommand{\lrgh}{\longleftrightarrow}

\newcommand{\rgh}{\rightarrow}

\let\a=\alpha \let\b=\beta   \let\d=\delta  \let\e=\varepsilon
 \let\g=\gamma     \let\k=\kappa  \let\l=\lambda
      \let\o=\omega    \let\p=\pi  
\let\r=\rho  \let\s=\sigma \let\t=\tau   
 \let\x=\xi \let\z=\zeta
     \let\L=\Lambda 
\let\O=\Omega      
\newcommand{\rosso}{\textcolor{black}}

\author[A.~Faggionato]{Alessandra Faggionato}
\address{Alessandra Faggionato.
  Dipartimento di Matematica, Universit\`a di Roma `La Sapienza'
  P.le Aldo Moro 2, 00185 Roma, Italy}
\email{faggiona@mat.uniroma1.it}

\title[Mott's  law]{
Mott's law for the critical conductance of Miller--Abrahams random resistor network}

\begin{document}

\begin{abstract}
In this short note we derive Mott's law for the critical conductance  of the  Miller--Abrahams random resistor network on a Poisson point process on $\bbR^d$, $d\geq 2$, and we give a percolative characterization of the factor preceding the temperature dependent term $\b^\frac{\a+1}{\a+1+d}  $. We also give mathematical arguments supporting its universality.  This note is a  preliminary version of a more extended work, where we also discuss the equality between the effective conductance of the resistor network and 
the critical conductance.

\medskip


\thanks{This work   has been  supported  by  PRIN
  20155PAWZB ``Large Scale Random Structures". }

\end{abstract}

\maketitle

\section{Introduction}
Mott's  variable range hopping is a mechanism of phonon--assisted electron transport taking place in amorphous solids (as doped semiconductors)  in the regime of strong Anderson localisation. It has been introduced by Mott in order to explain the anomalous 
non--Arrenhius decay of the conductivity at low temperature \cite{Mott}. We refer to \cite{MD,SE,POF} for a detailed discussion.
Keeping the language of doped semiconductors, in the regime of low impurity density some effective models  \rosso{have} been proposed. One can approximate  the localized electrons  by classical non--interacting particles moving according to random walks with  jump  probability  rates 
given by the electron transition rates multiplied  by a suitable factor mimicking the effect of  Pauli exclusion principle. 
The mathematical analysis of this random walk has  lead, between other, to the derivation of upper and lower bounds of the diffusion constant in agreement with Mott's law \cite{FM,FSS}.
Another effective model in the regime of low impurity density is given by the random  Miller--Abrahams resistor network \cite{MA} (shortly, MA resistor network), on which we concentrate here.

We describe the resistor network. The set of vertexes (nodes)  of the resistor nertwork   is a so called simple point process, i.e. a   random locally finite subset   $\xi\subset \bbR^d$.  We suppose the law of $\xi$ to be isotropic. 
Given a realization of $\xi$,  independently from the random mechanism generating $\xi$, one attaches to each vertex $x\in \xi$ a random variable $E_x$ called  \emph{energy mark}, in such a way that the energy marks $(E_x)_{x\in \xi}$  are i.i.d. random variables  with common law $\nu$. Physically, $E_x$ would be the ground state energy of the electron eigenfunction localized around the impurity $x$.  The physically relevant laws $\nu$  (in inorganic doped semiconductors)   are  of the form $c|E|^\a dE$ with support in some bounded  interval $[-C_0,C_0]$, $c$ being the normalizing constant and $\a$ be a nonnegative number:
\rosso{
\[
\nu (dE)=\frac{(\a+1)}{ 2 C_0 ^{\a+1} }  |E|^\a \mathds{1}( -C_0 \leq E \leq C_0)  dE
\]}

Then the MA resistor network is obtained by 
attaching  to any unordered pair of sites $x\not = y$ in $\xi$ a filament of conductivity  \cite{AHL,POF}
\begin{equation}\label{condu}
c_{x,y}:=\exp\bigl\{ - \frac{2}{\gamma} |x-y| -\frac{\b}{2} ( |E_x|+ |E_y|+ |E_x-E_y|) \}\,.
\end{equation}
Above $\g$ denotes the localization length  and $\b$ denotes the inverse temperature.
  Note that the skeleton of the resistor network  is the complete graph on $\xi$.  

As explained in \cite{AHL} one expects that the effective conductance of the MA resistor network is well approximated  by the critical conductance $c_c(\b)$ as $\b \to \infty$. To define it, given a number $c_*>0$ we denote by $\cG(c_*)$ the 
 graph  obtained from the MA resistor network by keeping filaments with conductivity at least $c_*$. Then the critical conductance $c_c (\b)$ is characterized by the following two properties:
(i) for any  value $c_*>c_c (\b)$  a.s. the graph  $\cG(c_*)$ 
 has no  unbounded cluster, (ii) for any value $c_*<c_c(\b)$  a.s. the graph  $\cG(c_*)$ 
 has some unbounded cluster.

 We will discuss the validity of the approximation of the effective conductance of the MA resistor network with the critical conductance $c(\b)$ in a future extension of this note. 
Our aim here is the derivation of Mott's law for the critical conductance $c_c(\b)$, which reads
\begin{equation}\label{mottino}
c_c(\b) \approx e^{- \k  \, \b^\frac{\a+1}{\a+1+d}}\,, \qquad \b \gg 1\,,
\end{equation}
for some $\b$--independent constant $\k >0$.
To this aim, we   sample the node set $\xi$ by a homogeneous Poisson point process (the resulting resistor network will be called Poisson MA resistor network). In this case we can provide a formula for $c_c (\b)$ (cf. Corollary \ref{uff}).    We will also give  a percolation  characterization  of the   factor $\k$ in \eqref{mottino} for Poisson MA resistor networks  and   provide arguments supporting the universality of the constant $\k$ outside the class of homogeneous Poisson point processes.

We point out that the analysis of the  factor $\k$ has lead to a certain debate in the physical literature with different proposals  \cite{POF}. The analysis carried on below  is indeed analytic and mathematically rigorous (hence, without any Ansatz).

\section{Poisson Miller--Abrahams resistor network 
}
In this section we suppose that the random set $\xi$ of the nodes of the MA resistor network is  given by 
 a homogeneous Poisson point process  with density $\r$ in $\bbR^d$. 
We denote by  
$\xi(A)$ the number of points in $\xi \cap A$ for $A \subset \bbR^d$. We recall that $\xi$ is a homogeneous Poisson point process  with density $\r $ in $\bbR^d$ if and only if the following two properties are satisfied:
\begin{itemize}
\item[(i)]  For mutually disjoint Borel subsets $A_1,A_2, \dots, A_k$ in $ \bbR^d$, the random variables $\xi(A_1), \xi(A_2),\dots, \xi(A_k)$ are independent;
\item[(ii)]  For any bounded  Borel subset $A\subset \bbR^d$, $\xi(A)$ is a Poisson random variable with parameter $\r\, \ell(A)$, $\ell(A)$ being the Lebesgue measure of $A$. This means that 
\[ \bbP\bigl(\xi (A)= k\bigr)= e^{-\l \ell(A)} \frac{ (\l \ell(A) ) ^k}{k!}\,, \qquad k=0,1,2,\dots\]
\end{itemize}
To simplify the notation and without any loss of generality, we take the localization length $\g$ equal to  $2$ and  we rename  $\beta/2$ by $\beta$ (keeping the name of inverse temperature), so that  the conductivity $c_{x,y}$ in \eqref{condu} now reads
\begin{equation}\label{conda}
c_{x,y}:=\exp\bigl\{ -  |x-y| -\b  ( |E_x|+ |E_y|+ |E_x-E_y|) \}\,.
\end{equation}

\smallskip

To study  the graph $\cG(c_*)$ described in the introduction, it is more convenient to write $c_*$ as $e^{-\z}$ for some $\z>0$. To stress the dependence of $\cG(c_*)$ on $\z,\b, \r$ we write $\cG_{\z,\b,\r}$:
\begin{Definition} 
Given  a threshold $\z >0$, the inverse  temperature $\b$ and the density $\rho$ of the Poisson point process,  the random graph 
$\cG_{\z, \b, \rho}=  \bigl( \xi , \cE _{\z, \b, \r}  \bigr)$  has vertex set $\xi $ and  edge set $\cE _{\z, \b, \r}  $  given by the pairs $\{x,y\} $ such that 
    $x\not =y$ in $ \xi $ and \begin{equation}
\label{connectcont}
|x-y|+\beta(|E_x|+|E_y|+|E_x-E_y)|)\leq \z\,.\end{equation} 
\end{Definition}
\bigskip 
We say that the graph $\cG_{\z, \b,\r}$ percolates if it contains an unbounded connected component.

\begin{Definition} We define  $\nu_\star$ as the probability distribution 
\[ \nu_\star (du)=  \frac{(\a+1)}{2}|u|^\a \mathds{1}(-1 \leq u \leq 1) du \,.\
 \] 
 The graph  $\cG_{\z, \b,\r}$ obtained when replacing  $\nu$ by $\nu_\star$ will be denoted 
 as $\cG^\star_{\z, \b,\r}$
\end{Definition}
\smallskip

\begin{Lemma}\label{rai}
There exists $\l_c >0$ such that if $\l <\l_c$ then a.s. the graph $\cG^\star _{1, 1,\l} $  does not percolate, while if $\l>\l_c$ then   a.s. the graph $\cG^\star_{1, 1,\l} $   percolates.
\end{Lemma}
\begin{proof}  Since two homogeneous Poisson point processes   with density $\l,\l'$  can be coupled in a way that the one with smaller
density is contained in the other, we get that the function $h(\l)$, defined as the probability that the graph $\cG^\star_{1, 1, \l}$  percolates, is weakly increasing. Hence, to get the thesis it is enough to exhibit two positive  constants $\l_1,\l_2$ such that $h(\l_1)=0$
 and $h(\l_2)>0$.

 Let us consider the graph $\cG^{(1)}_{\l}:=(\xi ,\cE^{(1)})$, where $\cE^{(1)}$ is the set of pairs $\{x,y\}$ of vertexes  of $\xi$ that satisfies $0<|x-y|\leq 1$. Then  $\cG^{(1)} _{\l}$ contains the graph $\cG^\star_{1, 1,\l} $ and can be seen as the realization of a Boolean model with deterministic radius $1/2$ on the Poisson point process $\xi$ with density $\l$ \cite{M}. Indeed two points $x,y$ are connected by an edge in $\cG^{(1)}_\l $ if and only if the closed balls centered at $x$ and $y$ with radius $1/2$ intersect.
It is known \cite{M}  that for $\l$ small   the graph $\cG^{(1)}_\l $ a.s.  does not percolate. Since $\cG_{1,1,\l}^\star\subset \cG^{(1)}_\l$, we conclude that $h(\l)=0$ for $\l$ small.

We now show that $h(\l)$ is positive for $\l$ large. To this aim  we fix $\d \in (0,1/3)$ and we consider the graph $\cG^{(2)}$ in $\bbR^d$
obtained by taking as vertex set $\cV^{(2)}:=\{ x\in \xi\,:\, E_x\in [0, \d]\}$ and taking as edge set
$\cE^{(2)}:= \bigl\{ \{x,y\}\,:\, x,y\in \cV^{(2)} \text{ and }  |x-y |\leq 1-3\d \bigr\}$. 
By construction, since $\z=\b=1$,    if $\{x,y\}\in \cE^{(2)}$ then the inequality \eqref{connectcont} is satisfied and therefore
$\{x,y\}\in \cE$. As a consequence 
 $\cG^{(2)}_\l$ is a subgraph of $\cG_{1,1,\l}^\star$. On the other hand, $\cG^{(2)}_\l$ can be thought of as a Boolean model associated to the deterministic radius $(1-3\d)/2$, on the Poisson point process $\cV^{(2)}$ with density \rosso{$\l \nu_\star([0,\d])$} (note that \rosso{$\nu_\star([0,\d])>0$}). It is known \cite{M} that  when this density is large, i.e. for  large $\l$,    the graph $\cG^{(2)}_\l $ a.s.   percolates and therefore the same happens for $\cG_{1,1,\l}^\star$ since $\cG^{(2)}_\l\subset \cG_{1,1,\l}^\star$. This proves that $h(\l)>0$ for $\l$ large.
 \end{proof}
\begin{Theorem} \label{teo1}  Let $\l_c$ be the universal constant appearing in Lemma \ref{rai}. 
Fixed $\b$ and $\rho$,
set
\begin{equation}
\label{zetacritico}
\z_c(\beta, \rho):= \bigl( {\l_c/\r}\bigr)^{\frac{1}{\a+1+d}}
  ( \rosso{ \b C_0}) ^{\frac{\a+1}{\a+1+d}}\,.
\end{equation}
\rosso{Then,  for $ \beta >    \l_c^\frac{1}{d} \rho ^{- \frac{1}{d}} C_0^{-1}$,
 the following holds: 
\begin{itemize}
\item[(i)] if $\z < \z_c(\b,\r)$, then 
a.s.   the graph $\cG_{\z, \b,\r}  $  does not percolate;
\item[(ii)]   if  $\z > \z_c(\b,\r)$, then a.s.   the graph $\cG_{\z, \b,\r}  $  percolates.
\end{itemize}
}
\end{Theorem}
\begin{proof} 
\rosso{The condition $\beta >   \l_c^\frac{1}{d} \rho ^{- \frac{1}{d}} C_0^{-1}$ is equivalent to the condition $ \z_c(\b,\r) < C_0 \b $. Under this condition 
 it is enough to prove the thesis for $\z  \leq C_0 \b $. Indeed, if  we prove the  thesis for $\z  \leq C_0 \b $, then we can proceed as follows for 
 $\z > C_0 \b$.  Since  $ \z  > \z_c(\b,\r)$, and we  fix 
  $\z' \in (\z_c(\b,\r), C_0 \b)$.  Then we would know that a.s.   the graph $\cG_{\z', \b,\r}  $  percolates. Since $\cG_{\z', \b,\r}  \subset \cG_{\z, \b,\r}  $, we conclude that 
  a.s.   the graph $\cG_{\z, \b,\r}  $  percolates.
  }
  
  \rosso{Due to the above claim from now on we restrict to $\z  \leq C_0 \b $}.
It is convenient to rewrite  \eqref{connectcont} as 
\begin{equation}
\label{connectcont2}
\frac{|x-y|}{\z}+\frac{\beta}{\z}\left(|E_x|+|E_y|+|E_x-E_y|\right)\leq 1.
\end{equation}
Hence, if $(\beta/\z) |E_x|>1$, then $x$ is an isolated point in the graph $\cG_{\z, \b,\r} $ and therefore  it does not give any contribution to   the existence of the infinite cluster. So
 $\cG_{\z, \b,\r} $ percolates if and only if the   graph 
  $\tilde \cG_{\z,\b,\r}=\bigl(\tilde \xi, \tilde \cE)$ percolates, where  $\tilde \cG_{\z,\b,\r} $ is defined as follows. Its vertex set is given
    by 
  \[\tilde \xi:= \{ x\in \xi\,:\, |E_x| \leq \z/\beta\}\,,
  \] 
  while its edge set $\tilde \cE$ equals $\cE _{\z, \b, \rho}$.
 We analyse in detail the graph   $\tilde \cG_{\z,\b,\r} $.\\

 \noindent
 $\bullet$ 
  The vertex set  $\tilde \xi$ is obtained by  thinning $\xi$. In particular, a point $x$ in $\xi$ survives in $\tilde \xi$ independently from the other points  with probability 
 \rosso{ \[
 \nu \bigl([ -\z/\b , \z/\b]\bigr) = \left( \frac{\z}{\b C_0 } \right) ^{\a+1} \,.
 \]
 }
\rosso{ As a consequence, $\tilde \xi$ is a homogeneous Poisson point process with density  $\tilde \rho :=(\z/\b C_0)^{\a+1} \rho$.} \\

 \noindent
 $\bullet$  Let us set $A_x:= (\b/\z)E_x$.
 \rosso{We claim that, knowing that $x \in \tilde \xi$, $A_x$ has distribution $\nu_\star$ defined in Lemma \ref{rai}.
 Indeed, given $0 \leq u \leq 1$,  
 \begin{equation}
 \begin{split}
 & P\left(  (\b/\z)E_x\in [0,u]\,\Big{|}\, |(\b/\z)E_x|\leq 1\right) =\frac{ P(  (\b/\z)E_x\in [0,u])}{ P(  |(\b/\z)E_x|\leq 1) }= \frac{ \nu ([0,(\z/\b) u])}{ \nu( [- \z/\b, \z/\b] )}\\
  &= \frac{(1/2)  \left( \frac{\z u }{\b C_0 } \right) ^{\a+1}    }{ \left( \frac{\z}{\b C_0 } \right) ^{\a+1}  } = \frac{1}{2} u^{\a+1}=\nu_\star (  [0,u])\,.
  \end{split}
  \end{equation}
 A similar result holds for what concerns the event $\{ (\b/\z)E_x\in [-u,0]\}$, thus the thesis.
 }
%
 
 \noindent
 $\bullet$ We observe that
 we can rewrite \eqref{connectcont2} in terms of the  variables $A_x$'s 
 as
 \begin{equation}
\label{connectcont3}
\frac{|x-y|}{\z}+|A_x|+|A_y|+|A_x-A_y|\leq 1\,.
\end{equation}

 Now we consider a third  graph $\hat \cG= (\hat \xi, \hat \cE)$ defined as follows. We first introduce the map $\phi:\bbR^d \to \bbR^d$ as $\phi(x)=x/\z$. Then we define $\hat\xi:=\phi (\tilde \xi)$ and we define 
 $ \hat \cE$ as the set of edges $\{ \phi (x) , \phi (y) \}$ with $\{x,y\}$ varying among $\tilde \cE$. To  any point $\phi (x) \in \hat \xi$ we  associate a new energy mark $B_{\phi(x) } := A_x$. Writing $\phi(x)=v$ and $\phi(y)=w$, we get that two vertexes  $v\not= w$ in $\hat \xi$ are connected by an edge in $\hat \cG$ if and only if 
  \begin{equation}
\label{connectcont4}
|v-w|+|B_v|+|B_w|+|B_v-B_w|\leq 1\,,
\end{equation}
  We know that the r.v.'s $B_v$'s are i.i.d. with law $\nu$. Moreover, since $\hat \xi$ is the $\phi$--image of the homogeneous Poisson point process $\tilde \xi$ with density \rosso{$(\z/ \b C_0 )^{\a+1} \rho$}, we get that $\hat \xi$ is a homogeneous Poisson point process with density \rosso{$\l:=(\z^{\a+1+d}/(\b C_0)^{\a+1})\rho$}.
  As a consequence of the above observations the random graph  $\hat \cG$ has the same law  of the random graph  \rosso{$\cG^\star_{1,1,\l}$}. Since $\cG_{\z, \b,\r} $ percolates whenever the graph $\hat \cG$ percolates (by construction), from Lemma \ref{rai} we deduce that  $\cG_{\z, \b,\r} $  a.s. does not percolate if $\l <\l_c$ and that $\cG_{\z, \b,\r} $  a.s. percolates if $\l >\l_c$.
  Trivially, the conditions $\l<\l_c$ 
  and $\l>\l_c$ equals the conditions $\z<\z_c(\b,\rho)$ and $\z>\z_c(\b,\rho)$ respectively, thus concluding the proof.
  \end{proof}
Recalling that we have written $c_*$ as $e^{- \z}$ we get the following result on the critical conductivity:
\begin{Corollary}\label{uff} Let $\l_c$ be the universal constant appearing in Lemma \ref{rai}.
Then the critical conductance  $c_c(\b)$ is given by 
\[ c_c(\b)= \exp\Big\{- \bigl( {\l_c/\r}\bigr)^{\frac{1}{\a+1+d}}
   (\b C_0)^{\frac{\a+1}{\a+1+d}}\Big\}\,.
   \]
\end{Corollary}

\section{Final remarks}
Let us suppose now that $\xi$ is an ergodic stationary simple point process on $\bbR^d$ with density $\rho$, i.e.
$\bbE[ \xi(B)]= \rho \ell(B)$ for any bounded Borel subset  $B \subset \bbR^d$. Then one could argue as in the proof of Theorem \ref{teo1} and conclude that the graph $\cG$ obtained from the Miller--Abrahams resistor network  by keeping filaments 
of conductivity at  least $e ^{-\z} $  percolates if and only if the graph $\hat \cG$ percolates, where $\hat \cG$ is obtained by keeping nodes $x$ with $|E_x| \leq \z/\b$ (the other nodes would be isolated in the   Miller--Abrahams resistor network) and  afterwards rescaling by  $\z$. By  \cite[Prop. 9.3.I]{DVJ} the nodes of the graph $\hat \cG$ behave asymptotically as a homogeneous Poisson point process, hence we expect that the low temperature asymptotics of the critical conductance  $c_c(\b)$ is the same of the Poisson asymptotics given in Corollary \ref{uff}. A more robust analysis of this universality will be given in a future work \cite{Fa}.


\end{document}